\PassOptionsToPackage{capitalize}{cleveref} 
\documentclass[a4paper,USenglish,cleveref]{lipics-v2021}

\hideLIPIcs  

\graphicspath{{figures/}}

\usepackage{cite}
\usepackage{xspace}
\usepackage{graphics}
\usepackage[ruled,linesnumbered,noend]{algorithm2e}
\usepackage{pgfplots}
\usepackage{dsfont}
\usepackage{braket}

\newcommand{\R}{\mathbb{R}}
\newcommand{\C}{\mathcal{C}}

\newcommand{\eps}{\varepsilon}

\newcommand{\fvd}{\mathrm{FVD}}
\newcommand{\freg}{\mathrm{freg}}
\newcommand{\gm}{\mathrm{GM}} 

\newcommand{\ThreeFaceVertex}[1]{
	\begin{tikzpicture}[scale=1]
		\def \r{#1};
		\def \R{4}
		\coordinate (p0) at (0,0);
		\coordinate (p1) at (-30:\R cm);
		\coordinate (p2) at (90:\R cm);
		\coordinate (p3) at (210:\R cm);
		\coordinate (p4) at (-30:\r*\R cm);
		\coordinate (p5) at (90:\r*\R cm);
		\coordinate (p6) at (210:\r*\R cm);
		\fill[color = blue!50] (p1) -- (p4) -- (p5) -- (p2) -- (p1 |- p2);
		\fill[color = green!50] (p2) -- (p5) -- (p6) -- (p3) -- (p3 |- p2);
		\fill[color = violet!90] (p3) -- (p6) -- (p4) -- (p1);
		\fill[color = red] (p4) -- (p5) -- (p6);
		\draw[very thick, color=black] (p1) -- (p4) ;
		\draw[very thick, color=black] (p2) -- (p5) ;
		\draw[very thick, color=black] (p3) -- (p6) ;
		\draw[very thick, color=black] (p4) -- (p5) ;
		\draw[very thick, color=black] (p5) -- (p6) ;
		\draw[very thick, color=black] (p6) -- (p4) ;
\end{tikzpicture}}
\newcommand{\ThreeFaceEdge}[1]{
	\begin{tikzpicture}[scale=1]
		\def \r{#1};
		\def \R{3}
		
		\coordinate (p0) at (0,0);
		\coordinate (bl) at (-\R,-\R);
		\coordinate (br) at (\R,-\R);
		\coordinate (tl) at (-\R,\R);
		\coordinate (tr) at (\R,\R);
		\coordinate (p1) at (0,-\R);
		\coordinate (p2) at (0,-\R*\r);
		\coordinate (p3) at (0,\R*\r);
		\coordinate (p4) at (0,\R);
		\fill[color = green!50] (bl) -- (p1) -- (p2) .. controls (-\R*\r,0) .. (p3) -- (p4) -- (tl);
		\fill[color = blue!50] (br) -- (p1) -- (p2) .. controls (\R*\r,0) .. (p3) -- (p4) -- (tr);
		\fill[color = red] (p2) .. controls (-\R*\r,0) .. (p3) .. controls ( \R*\r,0) .. (p2);
		\draw[very thick, color=black] (p1) -- (p2) ;
		\draw[very thick, color=black] (p2) .. controls (-\R*\r,0) .. (p3) ;
		\draw[very thick, color=black] (p2) .. controls ( \R*\r,0) .. (p3) ;
		\draw[very thick, color=black] (p3) -- (p4) ;
\end{tikzpicture}}
\newcommand{\ThreeFaceTwoFace}[1]{
		\begin{tikzpicture}[scale=1]
			\def \r{#1/100}
			\def \R{3}
			\coordinate (p0) at (0,0);
			\coordinate (bl) at (-\R,-\R);
			\coordinate (br) at (\R,-\R);
			\coordinate (tl) at (-\R,\R);
			\coordinate (tr) at (\R,\R);
			\fill[color = green!50] (bl) -- (br) -- (tr) -- (tl) -- (bl);
			\draw[draw = black, very thick, fill=red] (p0) circle (\R*\r);
		\end{tikzpicture}
}

\newcommand{\TwoFaceVertex}[1]{\Flip{#1}{4}}
\newcommand{\TwoFaceEdge}[1]{\Flip{#1}{3}}

\newcommand{\Flip}[2]{
	\begin{tikzpicture}[scale=1]
		\def \r{#1/100}
		\def \R{3}
		\def \ColorOne{green!50}
		\def \ColorTwo{red}
		\ifthenelse{#2 > 2}{\def \ColorThree{blue!50}}{\def \ColorThree{\ColorOne}}
		\ifthenelse{#2 > 3}{\def \ColorFour{violet!90}}{\def \ColorFour{\ColorTwo}}
		\coordinate (p0) at (0,0);
		\coordinate (bl) at (-\R,-\R);
		\coordinate (br) at (\R,-\R);
		\coordinate (tl) at (-\R,\R);
		\coordinate (tr) at (\R,\R);
		\ifthenelse{
			#1 < 0
		}{
			\coordinate (p1) at (\R*\r,0);
			\coordinate (p2) at (-\R*\r,0);
			\fill[color = \ColorOne] (bl) -- (p1) -- (tl);
			\fill[color = \ColorThree] (tr) -- (p2) -- (br);
			\ifthenelse{#2 = 4}{
				\fill[color = \ColorTwo] (tl) -- (p1) -- (p2) -- (tr);
				\fill[color = \ColorFour] (br) -- (p2) -- (p1) -- (bl);
				\draw[very thick, color=black] (p1) -- (p2) ;
			}{
				\fill[color = \ColorTwo] (tl) -- (p1) -- (bl) -- (br) -- (p2) -- (tr) -- (tl);
			}
			\draw[very thick, color=black] (bl) -- (p1) -- (tl) ;
			\draw[very thick, color=black] (br) -- (p2) -- (tr) ;
		}{
			\coordinate (p1) at (0,-\R*\r);
			\coordinate (p2) at (0,\R*\r);
			\fill[color = \ColorTwo] (tl) -- (p2) -- (tr);
			\fill[color = \ColorFour] (br) -- (p1) -- (bl);
			\ifthenelse{#2 > 2}{
				\fill[color = \ColorOne] (bl) -- (p1) -- (p2) -- (tl);
				\fill[color = \ColorThree] (tr) -- (p2) -- (p1) -- (br);
				\draw[very thick, color=black] (p1) -- (p2) ;
			}{
				\fill[color = \ColorOne] (tl) -- (bl) -- (p1) -- (br) -- (tr) -- (p2) -- (tl);
			}
			\draw[very thick, color=black] (bl) -- (p1) -- (br) ;
			\draw[very thick, color=black] (tl) -- (p2) -- (tr) ;
		}
\end{tikzpicture}}

\usepackage{mdframed}				

\newcommand{\deleted}[1]{}

\title{A Collapse Process for Farthest Voronoi Diagrams of Lines in Three Dimensions} 

\author{Evanthia Papadopoulou}
{Faculty of Informatics, Università della Svizzera italiana (USI), Lugano, Switzerland}
{evanthia.papadopoulou@usi.ch}
{https://orcid.org/0000-0003-0144-7384}{}

\author{Martin Suderland}
{Faculty of Informatics, Università della Svizzera italiana (USI), Lugano, Switzerland}
{martinsuderland.research@gmail.com}
{https://orcid.org/0000-0002-6604-6381}{}

\author{Zeyu Wang}
{Faculty of Informatics, Università della Svizzera italiana (USI), Lugano, Switzerland}
{zeyu.wang@usi.ch}
{https://orcid.org/0009-0004-4207-198X}
{} 

\authorrunning{E. Papadopoulou, M. Suderland, and Z. Wang} 

\Copyright{Evanthia Papadopoulou, Martin Suderland, and Zeyu Wang} 

\ccsdesc[500]{Theory of computation~Computational geometry} 

\keywords{Farthest Voronoi diagram, lines, three dimensions, trisector, structural properties, local extrema, collapse process, collapse algorithm} 

\funding{This research was supported by
	the Swiss National Science Foundation (SNF), project
        200021E$\_$201356. Martin Suderland was also supported by SNF project
        P500PT$\_$206736/1.} 

\nolinenumbers 

\EventEditors{Lin Chen and Nicole Megow}
\EventNoEds{2}
\EventLongTitle{37th International Symposium on Algorithms and Computation (ISAAC 2026)}
\EventShortTitle{ISAAC 2026}
\EventAcronym{ISAAC}
\EventYear{2026}
\EventDate{December 6--9, 2026}
\EventLocation{Hangzhou, China}
\EventLogo{}
\SeriesVolume{399}
\ArticleNo{59}

\begin{document}

\maketitle

\begin{abstract}
We study a \emph{collapse process} to construct the farthest Voronoi diagram of lines in three dimensions, given a spherical map of the diagram’s unbounded features. The collapse process sweeps through the diagram in order of decreasing distance from the farthest lines. It follows the \emph{shrinking map}, a cell complex on a topological sphere that encodes the locus of points with a fixed farthest distance. We show that, in three dimensions, the collapse process has exactly four non-terminal local event types that change the structure of the shrinking map: \emph{deletion, swap, local minimum}, and \emph{local maximum} events; plus one terminal event. This list is complete. 

We give intrinsic three-dimensional geometric descriptions of the four
non-terminal events. First, we classify the two vertex-related events, deletion
and swap events, by the spherical convex hull of the four tangent points from a vertex
to its four defining lines. Then, we analyze the events related to the local
extrema of the 
distance function along the trisector of three lines. We
show that the 
distance function along a trisector has at most $4$ local
maxima and $8$ local minima, and that both bounds are tight. The extrema can be
found via a polynomial of degree $12$. As a byproduct, this gives a direct
method for finding the smallest sphere tangent to three given lines.
At each local extremum, the tangent sphere touches the three lines at points lying on a great circle.

The collapse process and the completeness of the event list also apply, under similar general position assumptions, to the farthest Voronoi diagram of convex sites under strictly convex distance functions.
\end{abstract}

\section{Introduction}

Voronoi diagrams are versatile 
space partitioning structures
that encode proximity information 
among a set $S$ of simple geometric objects in a space, called sites.
The \emph{nearest} (respectively, \emph{farthest}) Voronoi diagram
of $S$ partitions the underlying space into regions 
that have the same closest (resp., farthest) site.
In this paper, we focus on the  Euclidean farthest Voronoi diagram of $n$ lines
in 3D, and describe a general paradigm
for its construction,  after the 
unbounded features of the diagram are known.

In the Euclidean plane, Voronoi diagrams have been intensively studied 
and they are 
generally well-understood~\cite{Aurenhammer2013}.
Early attention was given to the Euclidean Voronoi diagram of points
in $d$-dimensional space,
where the bisectors are hyperplanes, and the diagram has been thoroughly studied.
For fixed dimension $d$, a tight bound on the complexity of the Euclidean point Voronoi diagram is $O(n^{\lceil \frac{d}{2} \rceil})$ and 
it can be computed in $O(n \log n + n^{\lceil \frac{d}{2} \rceil})$ time~\cite{Edelsbrunner1986,Chazelle1991,Klee1980}.
These results also hold for certain polyhedral
norms~\cite{Boissonnat1998,Icking2001}, and for the farthest Voronoi diagram~\cite{Seidel1987}.

For sites more general than points, however, or for general metrics in $\R^d$ where the bisectors
are typically curved, Voronoi diagrams have been far less understood. 
Interpreting a Voronoi diagram in $\R^d$ as an \emph{arrangement} of distance functions of the given sites in $\R^{d+1}$~\cite{Edelsbrunner1986} is a common concept.
This implies that the nearest and farthest Voronoi diagrams have $O(n^{d+\epsilon})$ complexity,
for any~$\eps > 0$~\cite{Sharir1994}, assuming that the distance functions are simple enough; 
then they can be constructed in $O(n^{3+\epsilon})$
time~\cite{Agarwal1994} for $d=3$.
Taking this route, 
an algorithm for computing the Voronoi diagram of lines in  $\R^3$, using the envelope package of CGAL,
has been illustrated in~\cite{Hemmer2010}.
Some tighter combinatorial bounds are known for restricted cases in
$\R^3$, which include (near-)quadratic complexity bounds for various sites involving the Euclidean distance~\cite{Aronov2002,Koltun2002,Aurenhammer2017,Barequet2023}, as well as for convex distance functions 
induced by a polyhedron of constant complexity~\cite{Chew1998,Koltun2002a,Aurenhammer2021}.

Due to algebraic complexity, the Euclidean Voronoi diagram of lines,
or other 
generalized sites in 3D, remains a challenging problem.
Entire papers have been published for describing the Euclidean Voronoi diagram 
of only three~\cite{Everett2009} and recently four~\cite{papadopoulou2026} lines as base cases.
Unless properties of these diagrams are better understood, there
is little hope to progress on the well-known complexity gap 
between the $\Omega(n^2)$ and $O(n^{3+\eps})$ bounds~\cite{Mitchell2001}.

The unbounded features of the farthest (and order-$k$) Voronoi diagrams of lines and line segments 
in $\R^d$ were
studied in~\cite{Barequet2023}.
These features are encoded in the \emph{Gaussian map}, a map on the sphere of
directions, which has complexity $\Theta(n^{d-1})$ for the farthest diagram.
For the farthest Voronoi diagram in $\R^3$, the Gaussian map
can be computed in $O(n^2)$ time;
no algorithm to compute the farthest
Voronoi diagram was given in~\cite{Barequet2023}.
The possibility of an inward \emph{collapsing process} to construct the diagram, starting
at the Gaussian map, was proposed in~\cite{Barequet2013}.

In this paper,
we study the ``collapse process'' to construct the farthest Voronoi
diagram of lines in $\R^3$.
This process 
is widely used in two dimensions to construct tree-like Voronoi
diagrams, and we extend it to 3D for farthest Voronoi diagrams whose cells are 
unbounded.
Through this process, we gain more insight into the diagram's properties. 
The collapse process 
is of interest for diagrams whose cells are 
unbounded. 
It starts from the boundary
features, and maintains an active front while gradually discovering the entire
diagram inwards. 
It has been used in the areas of e.g.,
straight skeletons~\cite{Aichholzer1996}, farthest Voronoi
diagrams~\cite{Aurenhammer2006}, Brocard illumination of convex
polygons~\cite{Alegria2026}.

\subparagraph*{Contribution.}
We focus on the \emph{collapse process} for constructing the farthest Voronoi diagram of lines in $\R^3$. The process starts from the unbounded part of the diagram, encoded by a two-dimensional spherical map related to the Gaussian map~\cite{Barequet2023}. It then sweeps the upper envelope of the distance functions to each line in decreasing order. The level set of the upper envelope at any given distance value forms a labelled cell complex on a topological sphere, called the \emph{shrinking map}. 

We show a complete classification of the local structural changes of the
shrinking map during the collapse process.
In particular, there are exactly four non-terminal event types: \emph{deletion
  events, swap events, local minimum events}, and \emph{local maximum
  events}. Together with one terminal event, this list is complete. This
is in contrast  
to the planar collapse process for farthest diagrams of line segments,
which is 
governed by only one non-terminal event type similar to the deletion event~\cite{Aurenhammer2006}.
Among the four event types, the first two occur at a Voronoi vertex, and the
latter two occur at a local extremum 
along a trisector of three lines. 
%
We give an intrinsic three-dimensional classification of the vertex events by considering the tangent points from a vertex to its four defining lines, which lie on a common sphere centered at the vertex. If the spherical convex hull of these points is a triangle, the vertex gives a deletion event; if it is a quadrilateral, the vertex gives a swap event. This replaces the four-dimensional upper-envelope viewpoint by a purely three-dimensional (simpler) criterion.

For the 
distance function along a trisector, we prove that it has at most $4$ local
maxima and at most $8$ local minima, and that both bounds are tight. We also
show that all such extrema are algebraically computable by solving polynomials
of degree $12$.
As a byproduct, this gives a direct
method for finding the smallest sphere tangent to three given lines: compute the
local minima of the common distance along their trisector and choose the
minimum.
We also show that a local extremum is the center of a sphere that touches the three lines at points lying on a great circle. 

Algorithmically, the collapse process gives a 
framework for constructing
the farthest Voronoi diagram of a set of lines $L$, $\mathrm{FVD}(L)$.
By generating the relevant local-maximum events on the fly,
the resulting algorithm has time complexity 
$O(|\mathrm{FVD}(L)|\,k_{\max}\log n)$, after initialization, where $k_{\max}$ is the maximum face
size encountered during the collapse.
The initialization can take quadratic time in the number of unbounded faces
that bound each cell of the diagram, 
if implemented by brute force.
The removal of the additional dependence on $k_{\max}$ gives  a direction for future work.

The collapse viewpoint and the completeness of the event list also extend to
more general settings with convex sites and convex distance functions.


\section{Preliminaries}

Let $L=\{\ell_1,\dots,\ell_n\}$ be a set of $n$ lines in general position in
$\R^3$. For a point $p\in\R^3$ and a line $\ell\in L$, let $d(p,\ell)$ denote
the Euclidean distance from $p$ to $\ell$.
The \emph{farthest Voronoi region} of $\ell$ is
$\freg(\ell,L)=\{p\in\R^3 \mid \forall \ell'\in L\setminus \{\ell\},\ d(p,\ell)
> d(p,\ell')\}$.

The farthest Voronoi regions partition $\R^3$ into a cell complex, called the
\emph{farthest Voronoi diagram} of $L$, denoted $\fvd(L)$.
The diagram features are vertices, edges, 2D faces, and 3D cells.
A point of $\fvd(L)$ that is equidistant from exactly $k$ lines belongs to the
\emph{$k$-sector} of these lines, which is the locus of points equidistant from
$k$ lines. In particular, a $2$-sector is a \emph{bisector}, denoted by
$B(\cdot, \cdot)$, and a $3$-sector is a \emph{trisector}, denoted by $T(\cdot,
\cdot, \cdot)$.
Two trisectors are said to be \emph{related} if they involve exactly four distinct lines.

We assume the same general position assumptions as in~\cite{papadopoulou2026}: (1) lines are pairwise skew, and no three lines are parallel to a common
plane; (2) no sphere can be tangent to four lines with coplanar tangency points; (3) no sphere
is tangent to five lines; (4) no four lines have direction vectors whose representations on the
sphere of directions are cocircular. 
The general position assumption implies the following: all $k$-sectors have dimension $4{-}k$, Voronoi vertices of $\fvd(L)$ are incident to $4$ Voronoi edges, and related trisectors intersect transversely. 

The unbounded features of the diagram can be encoded as a two-dimensional
spherical map in two closely related ways as defined below.

\begin{definition} [\cite{Barequet2023}]
The Gaussian map of a cell complex $M$, denoted by $\gm(M)$, maps each cell in $M$ to its unbounded directions, which are encoded on the unit sphere $\mathbb{S}^2$. 
\end{definition}
The Gaussian map of the farthest Voronoi diagram of lines has complexity
$O(n^2)$ and can be computed in 
$O(n^2)$ time~\cite{Barequet2023}.

\begin{definition}[\cite{papadopoulou2026}]
Given a cell complex $M$, let $B$ be a ball large enough to intersect any
cell of $M$ in one connected component and let $\Gamma=\partial B$ be the
boundary sphere. The intersection of $M$ with $\Gamma$ is denoted by
$\Gamma(M)$, called the $\Gamma$-map of $M$.
\end{definition}
$\gm(M)$ is the limit of $\Gamma(M)$ as the radius of $\Gamma$ goes to $\infty$.
The $\Gamma$-map can be easily constructed from the Gaussian map, which will be shown in
Section~\ref{sec:outlook}. Next, we list some known facts on the farthest
Voronoi diagram of lines.

\begin{lemma}[Ray property~\cite{Barequet2023}]\label{lem:raysinfreg}
Let $p \in \freg(\ell, L)$ be a point in the
farthest Voronoi region of a line~$\ell\in L$, and let $t\in \ell$ be the point
on $\ell$  closest to~$p$.
Then, the ray that emanates from~$p$ with direction~$\overrightarrow{tp}$ is
entirely contained in~$\freg(\ell, L)$; further, the
distance to $\ell$ along this ray is strictly monotonically increasing.
\end{lemma}

\begin{corollary}[~\cite{Barequet2023}]\label{cor:unboundedfarthestcells}
All 3D cells of the farthest Voronoi diagram of lines are unbounded.
\end{corollary}

\begin{lemma}[\cite{papadopoulou2026}]\label{lem:cellcount}
In $\fvd(L)$, the farthest Voronoi region of each line has exactly $n{-}1$ many 3D cells. 
\end{lemma}


\section{The Collapse Process and Five Event Types}\label{sec:collapse}

We describe the collapse process for the farthest Voronoi diagram of a set $L$ of lines in~$\R^3$. Informally, let the priority of a point be the farthest distance from a line in $L$. The collapse process scans the farthest diagram in decreasing priority, starting from infinity. During the process, it maintains a labelled cell complex, called the \emph{shrinking map} that contains all points of the same priority. As the priority decreases, the shrinking map maintains its structure except at a few discrete critical values, which are called the \emph{events} of the collapse process.

Formally, for a point $x \in \mathbb{R}^3$, define its \emph{priority} as $\pi(x) = \max_{\ell \in L} d(x, \ell)$.
If $x$ lies in a feature (vertex, edge, or face) of $\fvd(L)$, then $\pi(x)$ is the (common) distance from $x$ to each of its defining lines. 

\begin{definition}[Shrinking map]\label{def:shrinking-map}
For $\lambda \in \mathbb{R}$, the \emph{shrinking map} at priority $\lambda$ is
$M_\lambda = \{x \in \mathbb{R}^3 \mid \pi(x) = \lambda\}$. Each point $x\in M_\lambda$ is labelled by the set
$\{\ell\in L\mid d(x,\ell)=\pi(x)\}$
of lines that realize the farthest distance from $x$.
The labeling induces a cell complex, where the points of one cell have the same
label.
\end{definition}

In other words, the shrinking map $M_\lambda$ is the intersection of the upper envelope of the graphs of the distance functions $d(x, \ell_i)$ with the hyperplane $x_4=\lambda$.

\begin{lemma}\label{lem:shrinking-sphere}
Let $\lambda_0=\min_{x\in\R^3}\pi(x)$. For every $\lambda>\lambda_0$,
the shrinking map $M_\lambda$ is a topological sphere. For every $\lambda<\lambda_0$,
it is empty.
\end{lemma}

\begin{proof}
Since $d(x, \ell)$ is convex for each $\ell$, $\pi(x)$ is also convex. Moreover, $\lim_{\|x\|\to\infty} \pi(x)=\infty$. Hence, for every $\lambda$, the sublevel set $K_\lambda=\{x\in\R^3\mid \pi(x)\le \lambda\}$
is compact and convex. Further, $M_\lambda=\partial K_\lambda$.
If $\lambda<\lambda_0$, then $K_\lambda=\emptyset$, and thus
$M_\lambda=\emptyset$. If $\lambda>\lambda_0$, then $K_\lambda$ has nonempty interior. Thus, $K_\lambda$ is homeomorphic to a ball and $M_\lambda$ is homeomorphic to $\mathbb S^2$.
\end{proof}

Consequently, for all values $\lambda$ such that $M_\lambda$ is nonempty, we can view the shrinking map as a planar cell complex on the sphere. Its faces (resp.\ edges and vertices) correspond to farthest Voronoi regions (resp.\ faces and edges).

For sufficiently large $\lambda$, the shrinking map $M_\lambda$ is topologically
equivalent to the $\Gamma$-map of $\fvd(L)$.
The collapse process starts from this particular cell complex and decreases
$\lambda$.

\subsection{Events of the Collapse Process}

In this section, we show that the collapse process for the farthest Voronoi diagram of three-dimensional lines has only five types of events: four non-terminal events and a terminal event. For comparison, the analogous process for the 2D farthest Voronoi diagram of line segments has only one non-terminal event. 
We first list the event types. Figure~\ref{fig:events} shows the effects of the events by depicting the local pieces of the shrinking map before/at/after the events.

\begin{figure}[!htbp]
\centering
\begin{tabular}{m{2.6cm} | m{2.4cm} | m{2.4cm} | m{2.4cm}}
    \textbf{Event name} & \textbf{Before} & \textbf{At} & \textbf{After} \\ \hline \hline
    \textbf{Deletion} & 
    \resizebox{0.10\textwidth}{!}{\ThreeFaceVertex{0.5}} & 
    \resizebox{0.10\textwidth}{!}{\ThreeFaceVertex{0.05}} & 
    \resizebox{0.10\textwidth}{!}{\ThreeFaceVertex{0}} \\ \hline
    \textbf{Swap} & 
    \resizebox{0.10\textwidth}{!}{\TwoFaceVertex{50}} & 
    \resizebox{0.10\textwidth}{!}{\TwoFaceVertex{0}} & 
    \resizebox{0.10\textwidth}{!}{\TwoFaceVertex{-50}} \\ \hline
    \textbf{Local minimum} & 
    \resizebox{0.10\textwidth}{!}{\ThreeFaceEdge{0.5}} & 
    \resizebox{0.10\textwidth}{!}{\ThreeFaceEdge{0.05}} & 
    \resizebox{0.10\textwidth}{!}{\ThreeFaceEdge{0}} \\ \hline
    \textbf{Local maximum} & 
    \resizebox{0.10\textwidth}{!}{\TwoFaceEdge{-50}} & 
    \resizebox{0.10\textwidth}{!}{\TwoFaceEdge{0}} & 
    \resizebox{0.10\textwidth}{!}{\TwoFaceEdge{50}} \\ \hline
    \textbf{Terminal event} & 
    \resizebox{0.10\textwidth}{!}{\ThreeFaceTwoFace{50}} & 
    \resizebox{0.10\textwidth}{!}{\ThreeFaceTwoFace{3}} & 
    end \\ \hline
\end{tabular}
\caption{List of all collapse events. The columns 2-4 show how the shrinking map changes locally at each event, where a local piece of the shrinking map is depicted before, at, and after the event. Different colors represent faces of different labels.}\label{fig:events}
\end{figure}

\subparagraph*{Deletion event and swap event.}
These are two events associated with a Voronoi vertex. At a deletion event, a shrinking-map face bounded by three edges disappears at a Voronoi vertex. At a swap event, no map face is created or deleted; one map edge disappears, and a different pair of map faces becomes adjacent along a new map edge.

\subparagraph*{Local minimum event and local maximum event.}
At a local minimum event, a shrinking-map face bounded by two edges disappears. At a local maximum event, two edges incident to a common shrinking-map face meet and split the face into two.
They occur, respectively, at a local minimum and a local maximum of the common distance function along a trisector of three lines.

\subparagraph*{Terminal event.}
This event shrinks two faces at the same time, after which the shrinking map is empty. In Lemma~\ref{lem:stop-final}, we show if it occurs, it is the last event of the collapse process. 

\begin{theorem}[Completeness of the event list]\label{thm:complete_event_list}
Let $L$ be a set of lines in $\R^3$ in general position. During the collapse
process, with decreasing priority $\lambda$,
between any two consecutive event priorities, the shrinking map $M_\lambda$ remains
 the same combinatorially.  At any event priority, the change is one of the five types listed above: deletion event, swap event, local minimum event, local maximum event, or terminal event. 
\end{theorem}

In the remainder of this section we prove Theorem~\ref{thm:complete_event_list}.
In Sections~\ref{sec:vertex-events} and~\ref{sec:trisector-extrema}, we characterize vertex events and local extremum events respectively in three-dimensional space.


\subsection{Completeness of the Event List}\label{sec:event-completeness}

We prove \cref{thm:complete_event_list} by ruling out all other local changes of the shrinking map. We organize the proof through a sequence of structural observations in Proposition~\ref{prop:collapse}. 
Recall the ray property from Lemma~\ref{lem:raysinfreg}. It gives a map from lower priority to higher priority within the same farthest region. This leads to the following observations.

\begin{proposition}\label{prop:collapse}
	During the collapse process, as $\lambda$ decreases, the following hold for the shrinking map $M_\lambda$. 
	\begin{enumerate}
		\item (No spontaneous face creation). No face can appear without an ancestor at higher priority.
		\item (No face merge). Distinct faces of $M_\lambda$ labelled by the same line cannot merge.
		\item (No new holes). No new holes can appear on an existing face.
	\end{enumerate}
\end{proposition}

\begin{proof}
For a line $\ell\in L$, let $A_{\ell}(\lambda) = M_{\lambda}\cap \freg(\ell, L)$; then its connected components are the 
faces of $M_\lambda$ labelled by $\ell$. Throughout the proof, let $\lambda_1<\lambda_2$ be two values in the shrinking process. Consider a point $p\in A_\ell(\lambda_1)$. By Lemma~\ref{lem:raysinfreg}, the ray $r(p)$ emanating from $p$ away from $\ell$ is contained in $\freg(\ell,L)$, and the priority along it is increasing. Hence, $r(p)$ intersects $M_{\lambda_2}$ at a unique point, denoted by $f(p)$. This gives a continuous function
$f: A_\ell(\lambda_1) \to A_\ell(\lambda_2)$.

\emph{No spontaneous face creation.} Let $p\in A_\ell(\lambda_1)$, then $f(p)\in A_\ell(\lambda_2)$. Hence, every point of a face of $\ell$ at priority $\lambda_1$ has an ancestor in the same farthest region at a larger priority $\lambda_2$.

\emph{No face merges.} Assume that two distinct faces $f_1$ and $f_2$ of a line $\ell$ on $M_{\lambda_2}$ merge at a smaller priority $\lambda_1$. Choose points $p_1\in f_1$ and $p_2\in f_2$. By the merge, there is a path $\xi:p_1\to p_2$ contained in $\freg(\ell,L)$ and consisting only of points with priority at most $\lambda_2$.
Consider the set $f(\xi)\subseteq A_\ell(\lambda_2)$, with $f(p_1)=p_1$ and $f(p_2)=p_2$. Since $f$ is continuous, $f(\xi)$ is a path in $A_\ell(\lambda_2)$ connecting $f_1$ and $f_2$, contradicting that $f_1$ and $f_2$ are distinct faces. See Figure~\ref{fig:mergeface} for an illustration.

\begin{figure}[h]
	\centering
	\includegraphics[scale=0.75]{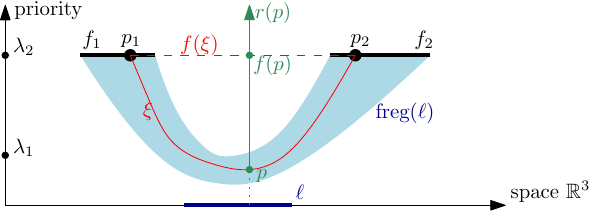}
	\caption{Impossibility of merging two faces during the collapse process.}\label{fig:mergeface}
\end{figure}

\emph{No new holes.} Let $D$ be a connected component of $A_\ell(\lambda_1)$. 
First, we show that the map $f$ is injective: two points on $M_{\lambda_1}$ that map to the same point of $M_{\lambda_2}$ must lie on the same ray orthogonal to $\ell$, and such a ray contains exactly one point at distance $\lambda_1$ from $\ell$. The inverse map sends a point of $f(D)$ to the unique intersection of the corresponding segment toward $\ell$ with $M_{\lambda_1}$, thus is continuous. Hence, $f$ is a homeomorphism from $D$ onto $f(D)$.
Thus, all holes of $D$ are already present in $f(D)$ at the larger priority $\lambda_2$.
\end{proof}

\begin{lemma}\label{lem:stop-final}
If the terminal event occurs during the collapse process, it is the last event. Further, it corresponds to a global minimum of the common distance on a bisector.
\end{lemma}
\begin{proof}
Let the terminal event be defined by two lines $\ell_i$ and $\ell_j$. Let $\delta=d(\ell_i,\ell_j)$. Let $m$ be the midpoint of the segment realizing the distance between the two lines, then $d(m, \ell_i) = d(m, \ell_j) = \delta / 2$.
Moreover, for every point $p\in\mathbb{R}^3$,
$\max\{d(p, \ell_i), d(p, \ell_j)\} \geq \delta / 2$.

The terminal event occurs at $\lambda=\delta/2$. For every $p\in\mathbb{R}^3$,
\(\pi(p) = \max_{\ell \in L} d(p, \ell) \geq \max\{d(p, \ell_i), d(p, \ell_j)\} \geq \delta / 2\).
Thus, no point has priority smaller than the priority of the terminal event. Thus, the terminal event is the last event in the collapse process.
\end{proof}

We are now ready to prove Theorem~\ref{thm:complete_event_list}. After Proposition~\ref{prop:collapse} and Lemma~\ref{lem:stop-final}, only local deletion, adjacency exchange, face split, and terminal event remain during the process. 

\begin{proof}[Proof of \cref{thm:complete_event_list}]
Consider a critical value of $\lambda$ at which the shrinking map $M_\lambda$ changes as a labelled cell complex. We distinguish the possible local changes.

First suppose that a face disappears. If a face $f$ is incident to 4 or more other faces, then the deletion of this face would correspond to a point equidistant to at least 5 sites, contradicting the general position assumption~(3). 
A face incident to 3 (resp.\ 2 or 1) faces can be deleted as part of the deletion (resp.\ local minimum or stop) event.

Next, we show that the following ``touch event'' (as depicted in Figure~\ref{fig:touch}) cannot happen. If it happens, then two related trisectors (the red-blue-green and red-blue-purple trisectors in Figure~\ref{fig:touch}) intersect tangentially, which contradicts the general position assumption. 

\begin{figure}[!htbp]
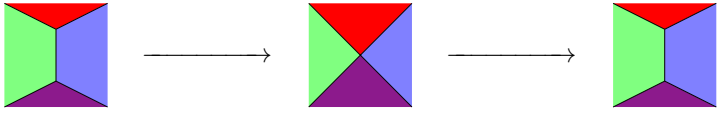

\centering
\[
\vcenter{\hbox{\resizebox{0.10\textwidth}{!}{\TwoFaceVertex{50}}}}
\quad
\xrightarrow{\hspace{1.5cm}}
\quad
\vcenter{\hbox{\resizebox{0.10\textwidth}{!}{\TwoFaceVertex{0}}}}
\quad
\xrightarrow{\hspace{1.5cm}}
\quad
\vcenter{\hbox{\resizebox{0.10\textwidth}{!}{\TwoFaceVertex{50}}}}
\]
\caption{The ``touch event'': a local piece of $M_\lambda$ is depicted before, at, and after the event. }\label{fig:touch}
\end{figure}

Consequently, when no face disappears and an adjacency changes, the old adjacency is replaced by a new one. The only possibility is a swap event which occurs at a Voronoi vertex.

Now suppose that a face splits. The split involves two boundary edges of the split face, which is exactly a local maximum event. All other local modifications are excluded by Proposition~\ref{prop:collapse}. This completes the proof.
\end{proof}


\section{Intrinsic Classification of Vertex Events}\label{sec:vertex-events}
 
The collapse process distinguishes deletion and swap
vertices using the upper envelope in~$\mathbb{R}^4$.
However, the difference already lies in their intrinsic geometry
in~$\mathbb{R}^3$ as we show in Theorem~\ref{thm:spherical-vertex-classification}.
The criterion is simple: the four tangency points from a vertex to its defining lines form either a spherical triangular hull or a quadrilateral hull, corresponding respectively to a deletion or a swap event.

\subparagraph{Notation.}
Let $x$ be a vertex defined by four lines $L_4=\{\ell_1,\ell_2,\ell_3,\ell_4\}$. Write
$d_i(\cdot) = d(\cdot, \ell_i)$
and let
$g_i = \nabla d_i(x)$.
A normal vector to the bisector $B(\ell_i,\ell_j)$ at $x$ is
$g_i - g_j$.
For $i<j<k$, we locally orient the trisector $T(\ell_i,\ell_j,\ell_k)$ at $x$ by the tangent vector
$t_{ijk} = (g_i - g_j) \times (g_i - g_k)$. 

By locally parameterizing the trisector at $x$, the derivative of the common distance function at $x$ along the oriented trisector $T(\ell_i,\ell_j,\ell_k)$ is
$g_i \cdot t_{ijk}$.
Locally around the vertex $x$, the trisector is split into two local arcs. 
We call the \emph{positive side} (around the vertex $x$) the local arc on which the common distance increases. In other words, if $g_i\cdot t_{ijk}>0$, the positive side points to the direction $t_{ijk}$; otherwise it points to the direction $-t_{ijk}$.

The following test decides, for a trisector incident to $x$, whether its positive side locally belongs to $\fvd(L_4)$. Denote $\eta = \det(g_1, g_2, g_3) - \det(g_1, g_2, g_4) + \det(g_1, g_3, g_4) - \det(g_2, g_3, g_4)$.

\begin{lemma}[FVD-side test]\label{lem:fvd-side-test}
The positive sides of the four trisectors
$T(\ell_1,\ell_2,\ell_3)$, $T(\ell_1,\ell_2,\ell_4)$, $T(\ell_1,\ell_3,\ell_4)$, $T(\ell_2,\ell_3,\ell_4)$
belong to $\fvd(L_4)$ if and only if, respectively,
$\det(g_1, g_2, g_3)\eta > 0$, $-\det(g_1, g_2, g_4)\eta > 0$, $\det(g_1, g_3, g_4)\eta > 0$, $-\det(g_2, g_3, g_4)\eta > 0.$
\end{lemma}

\begin{proof}
  We prove the first condition; 
  the others are analogous. Consider the trisector $T(\ell_1,\ell_2,\ell_3)$ and the function
$h = d_1 - d_4$ restricted to it. Since $x$ is equidistant from the four lines, the side of $T(\ell_1,\ell_2,\ell_3)$ belongs locally to $\fvd(L_4)$ exactly when $h>0$ on that side. Along the direction $t_{123}$, the derivative of $h$ is 
$(g_1 - g_4) \cdot t_{123}$.
A direct calculation gives
\((g_1 - g_4) \cdot t_{123} = \eta\).
On the other hand, the positive side of $T(\ell_1,\ell_2,\ell_3)$ is the side selected by the sign of
$g_1 \cdot t_{123} = \det(g_1, g_2, g_3)$.
Hence, the positive side lies in $\fvd(L_4)$ if and only if
$\det(g_1, g_2, g_3)\eta > 0$.
\end{proof}

Consider the collapse process of the four lines in $L_4$.
Just above the priority of $x$, the positive sides of trisector arcs of $\fvd(L_4)$ incident to $x$ appear on the shrinking map (by definition of the positive side).
And just below the priority of $x$, the negative sides that belong to $\fvd(L_4)$ start to appear in the shrinking map. Therefore, at a vertex $x$, if three of the four conditions of Lemma~\ref{lem:fvd-side-test} hold, then $x$ is a deletion vertex. If two of them hold, then $x$ is a swap vertex. By Theorem~\ref{thm:complete_event_list}, these are the only two types of vertices. Hence, no other case can happen. This yields a compact algebraic test.

\begin{proposition}\label{prop:vertex-product-test}
The type of a vertex $x$ can be determined by evaluating \[\operatorname{sgn}(\det(g_1, g_2, g_3) \det(g_1, g_2, g_4) \det(g_1, g_3, g_4) \det(g_2, g_3, g_4))\] at $x$. It is $+1$ if and only if $x$ is a swap vertex and $-1$ if and only if $x$ is a deletion vertex.
\end{proposition}

\begin{proof}
Multiplying the four expressions in Lemma~\ref{lem:fvd-side-test} shows that the sign of the displayed product records the number of sign conditions that hold modulo $2$. By the discussion above, two conditions hold for a swap vertex and three hold for a deletion vertex, giving signs $+1$ and $-1$, respectively.
\end{proof}

Though Proposition~\ref{prop:vertex-product-test} is already a test in 3D, it can be expressed more geometrically. Let $q_i$ be the closest point on $\ell_i$ to $x$, $i\in\{1,2,3,4\}$. Since $x$ is equidistant from the four lines, the four points $q_i$ lie on the sphere centered at $x$ with radius $d(x,\ell_i)$.

\begin{theorem}[Classification of vertex events]\label{thm:spherical-vertex-classification}
The type of a vertex $x$ is determined by the spherical convex hull of
$q_1,q_2,q_3,q_4$ on the sphere centered at $x$. If this hull is a spherical triangle, then $x$ is a deletion vertex. If it is a spherical quadrilateral, then $x$ is a swap vertex.
\end{theorem}

\begin{proof}
It is equivalent to consider the unit directions $g_i$'s on the unit sphere, since $q_i=x-rg_i$ for $i\in \{1,2,3,4\}$, where $r=d(x,\ell_i)$.
Assume first that the spherical hull is a triangle and~$g_1, g_2, g_3$ are vertices of the hull. Then the ray spanned by $g_4$ lies in the positive cone generated by $g_1,g_2,g_3$. Thus, for some $a,b,c>0$, we may write $g_4 = ag_1+bg_2+cg_3$.
Let $D=\det(g_1, g_2, g_3)$. Then
$\det(g_1, g_2, g_4) = cD$, $\det(g_1, g_3, g_4) = -bD$, $\det(g_2, g_3, g_4) = aD$.
Therefore,
\[\operatorname{sign}\left(\det(g_1, g_2, g_3)\det(g_1, g_2, g_4)\det(g_1, g_3, g_4)\det(g_2, g_3, g_4)\right) = \operatorname{sign}(-abcD^4) = -1.\]
By Proposition~\ref{prop:vertex-product-test}, $x$ is a deletion vertex. 

If the spherical hull is a quadrilateral, then after relabelling there are some $a,b,c,d>0$ such that $ag_1+cg_3=bg_2+dg_4$. An analogous calculation shows that $x$ is a swap vertex.  
\end{proof}

The next section gives the analogous intrinsic description for local extremum events.

\section{Local Extrema on Trisectors}\label{sec:trisector-extrema}

Every edge of $\fvd(L)$ is contained in the trisector of three lines. Besides
vertex events and the terminal event, the collapse process can change only at
local extremum events, which correspond to local extrema of the common distance
function along a trisector, see \cref{fig:trisectorfigure}.

\begin{figure}[h]
	\centering
	\includegraphics[width = 0.7\textwidth, trim = 120 40 0 280, clip]{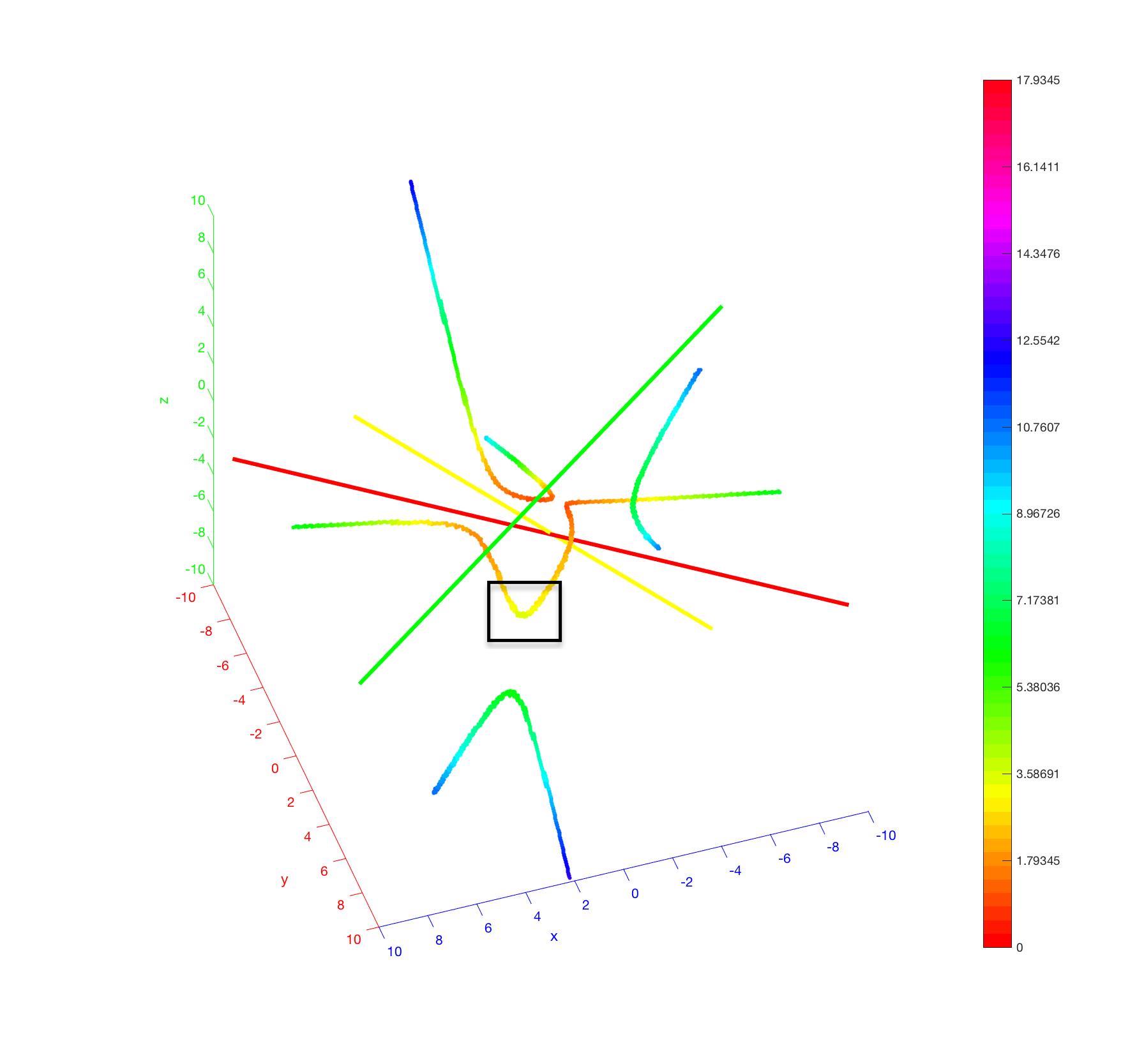}
	\caption{The trisector of three lines (red, green, yellow). The colors along the trisector encode the distance to the lines. In this example, the distance function along the middle branch admits a local maximum, as the color coding changes from orange to yellow to orange. The position of the local maximum is highlighted in the box. 
	}\label{fig:trisectorfigure}
\end{figure}

We prove three facts about these local extrema: tight bounds, a degree 12 algebraic computation, and a simple three-dimensional geometric interpretation.

Let
$T = T(\ell_1, \ell_2, \ell_3)$
be the trisector of three lines in $\mathbb{R}^3$. For $p \in T$, define
$\rho(p) = d(p, \ell_1) = d(p, \ell_2) = d(p, \ell_3)$.
By definition, local minima (resp.\ maxima) of $\rho$ correspond to local minimum (resp.\ maximum) events of the collapse process. 

\begin{theorem}\label{thm:trisector-extrema-count}
For three lines in general position in $\mathbb{R}^3$, the function $\rho$ on their trisector has at most $4$ local maxima and at most $8$ local minima. Both bounds are tight.
\end{theorem}
\begin{proof}
By the standard normal form for three skew lines from~\cite{Everett2009}, after a suitable choice of coordinates and scaling, one may assume that the three lines are parameterized as follows,
$\ell_1 : (0,0,1) + (1,a,0)t$, $\ell_2 : (0,0,-1) + (1,-a,0)t$,
and
$\ell_3 : (b_3,c_3,0) + (d_3,e_3,1)t$.
The bisector $B(\ell_1,\ell_2)$ is the set
$\{(x,y,z)\in\mathbb{R}^3 \mid z = -\frac{a}{1+a^2}xy\}$. Consider the projection
\[\Pi : B(\ell_1,\ell_2) \to \mathbb{R}^2, \quad (x,y,z) \mapsto (x,y),\]
which is a homeomorphism~\cite{Everett2009}. We aim to study the projected trisector $\Pi(T)$. 

Substituting $z = -\frac{a}{1+a^2}xy$
into the equation $d(p,\ell_1)=d(p,\ell_3)$ gives a bivariate polynomial equation
$P(x,y) = A(x)y^2 + B(x)y + C(x) = 0$,
where $A,B,C$ are quadratic polynomials in $x$ whose coefficients depend on the parameters $a,b_3,c_3,d_3,e_3$. Thus,
$\Pi(T) = \{(x,y) \in \mathbb{R}^2 \mid P(x,y) = 0\}$.

Let $d^2(x,y)$ be the squared distance from $\Pi^{-1}(x,y)$ to $\ell_1$. Since $\rho>0$, the local extrema of $\rho$ are exactly the local extrema of $d^2$ on the curve $P=0$. At a constrained critical point, the gradients of $d^2$ and $P$ are linearly dependent. Equivalently,
\[Q(x,y) = \frac{\partial d^2}{\partial x}\frac{\partial P}{\partial y} - \frac{\partial d^2}{\partial y}\frac{\partial P}{\partial x} = 0.\]

The polynomials $P(x,y)$ and $d^2(x,y)$ both have bi-degree $(2,2)$. Hence, $Q(x,y)$ has bi-degree $(3,3)$. B\'ezout's theorem shows that the system of equations $P=0$ and $Q=0$ has at most $2\times 3+3\times 2=12$ real roots. 

The trisector of three lines in general position consists of four unbounded branches~\cite{Everett2009}. Each branch has $\rho\to\infty$ at both ends. Since $\rho$ is continuous, local extrema alternate and each branch has exactly one more local minimum than local maximum. Since the total number of critical points is at most $12$, the four branches together have at most $4$ local maxima and at most $8$ local minima.
Three lines with the parameters 
$(a, b_3, c_3, d_3, e_3) = (1, -9, 2, 8, -18)$ have exactly 12 local extrema, see Figure~\ref{fig:trisector-12-extrema}. Hence, the bounds are tight.
\end{proof}

\begin{figure}[!h]
	\centering
	\includegraphics[width=0.34\textwidth]{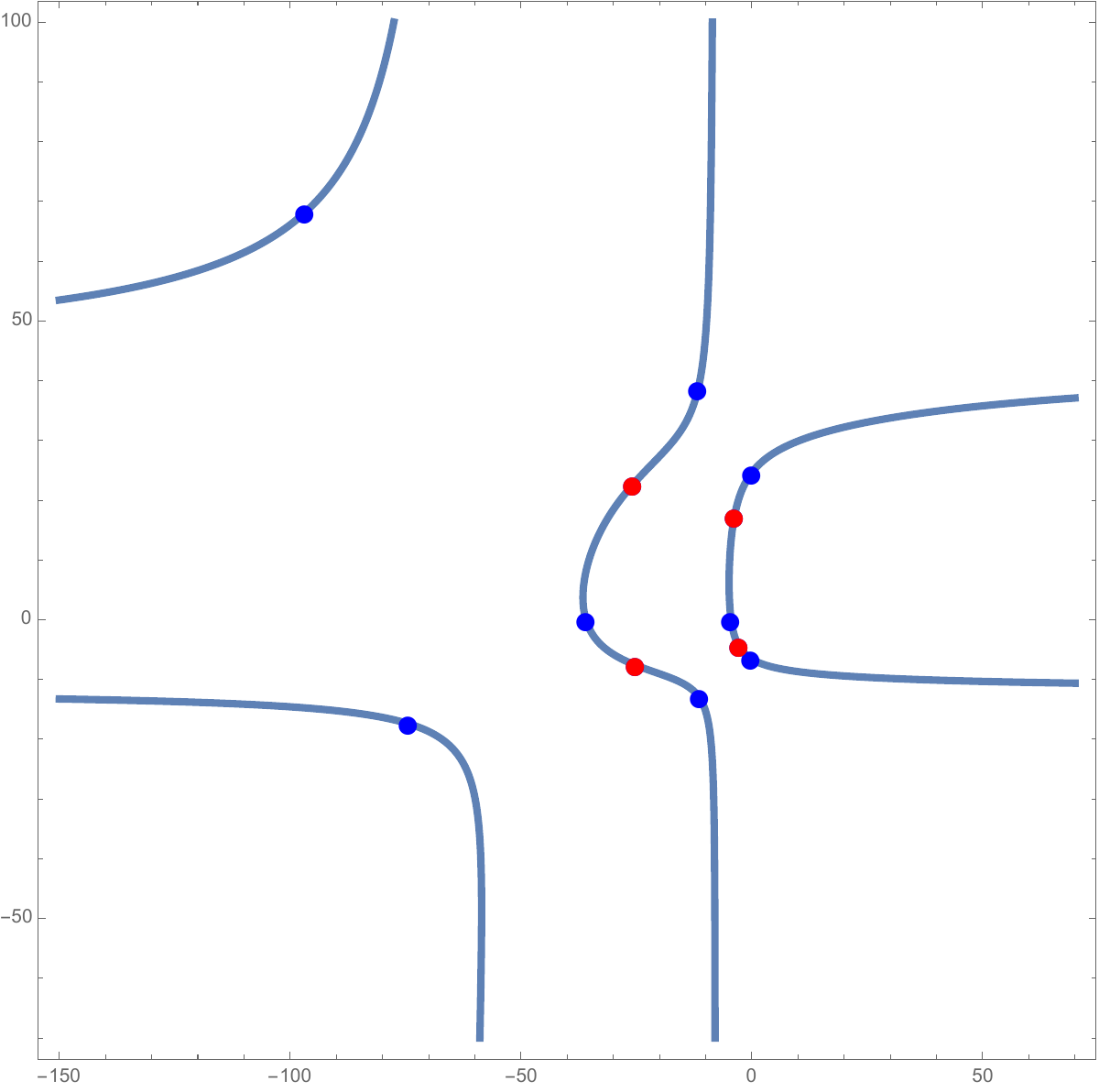}
	\caption{This figure shows the trisector of three lines $\ell_1, \ell_2, \ell_3$, with the parameters $(a, b_3, c_3, d_3, e_3) = (1, -9, 2, 8, -18)$, on the projected bisector $\Pi(B(\ell_1, \ell_2))$. The local minima are shown in \textcolor{blue}{blue} and local maxima are in \textcolor{red}{red}. There are 8 local minima and 4 local maxima.}
	\label{fig:trisector-12-extrema}
\end{figure}

\begin{proposition}\label{prop:smallest-touching-sphere}
Given three lines in general position in $\mathbb{R}^3$, the smallest sphere tangent to the three lines can be computed by solving a polynomial of degree $12$.
\end{proposition}
\begin{proof}
Normalize the three lines as in the proof of Theorem~\ref{thm:trisector-extrema-count}. Form the projected trisector equation $P(x,y)=0$ and the equation $Q(x,y)=0$. By the proof of Theorem~\ref{thm:trisector-extrema-count}, it suffices to find the real roots of the system $P=0 \land Q=0$. Eliminating $x$ from the system
by taking the resultant $\operatorname{Res}_x(P,Q)$ gives a univariate polynomial in $y$ of degree $12$. Solve this univariate polynomial in $y$ of degree $12$, and back substitute to obtain the $x$-coordinates. Then recover the 3D coordinates by $z = -\frac{a}{1+a^2}x y$.

This yields the centers $p=(x_0,y_0,z_0)$ of spheres tangent to the three lines. Every sphere tangent to the three lines has its center on the trisector, and $\rho$ tends to infinity on all unbounded ends of the trisector. Hence, the global minimum tangent sphere exists, and is attained at one of the local minima of $\rho$.
\end{proof}

Similar to the previous section on vertices, we also give an intrinsic 3D interpretation of the local extrema, in contrast to viewing them as critical points of the upper envelope in $\mathbb{R}^4$.

\begin{theorem}\label{thm:great-circle-extremum}
Let $x$ be a local extremum of $\rho$ on $T(\ell_1,\ell_2,\ell_3)$. Then the sphere centered at~$x$ with radius $\rho(x)$ touches $\ell_1,\ell_2,\ell_3$ in three points that lie on a common great circle.
\end{theorem}

\begin{proof}
Recall the definition of $q_i$ and $g_i$ introduced in Section~\ref{sec:vertex-events}: $q_i$ is the closest point on $\ell_i$ from $x$, and $g_i=\frac{x-q_i}{d(x,q_i)}$. To show that $q_i$'s lie on a common great circle of the sphere centered at $x$ with radius $\rho(x)$, it is equivalent to show that the $g_i$'s that lie on the unit sphere also lie on a great circle. By the same calculation as in Section~\ref{sec:vertex-events}, since $x$ is a local extremum of $\rho$, the derivative of the distance function vanishes at $x$. Hence, $g_1 \cdot t = 0$, where $t= (g_1 - g_2) \times (g_1 - g_3)$ is the tangent vector to the trisector at $x$. Expanding the expression gives $\det(g_1, g_2, g_3) = 0$, which holds if and only if the $g_i$'s lie on a common plane through the origin. Hence, they also lie on a great circle. This completes the proof. 
\end{proof}

It is worth noting that for line segments the trisector may contain a closed
branch, which is illustrated in  Figure~\ref{fig:bounded-trisector-segments}.
For such a trisector branch, it is obvious that the distance function must have local maximum and a local minimum. 

\begin{figure}
	\centering
	\includegraphics[width = 0.7\textwidth, trim = 0 0 0 0, clip]{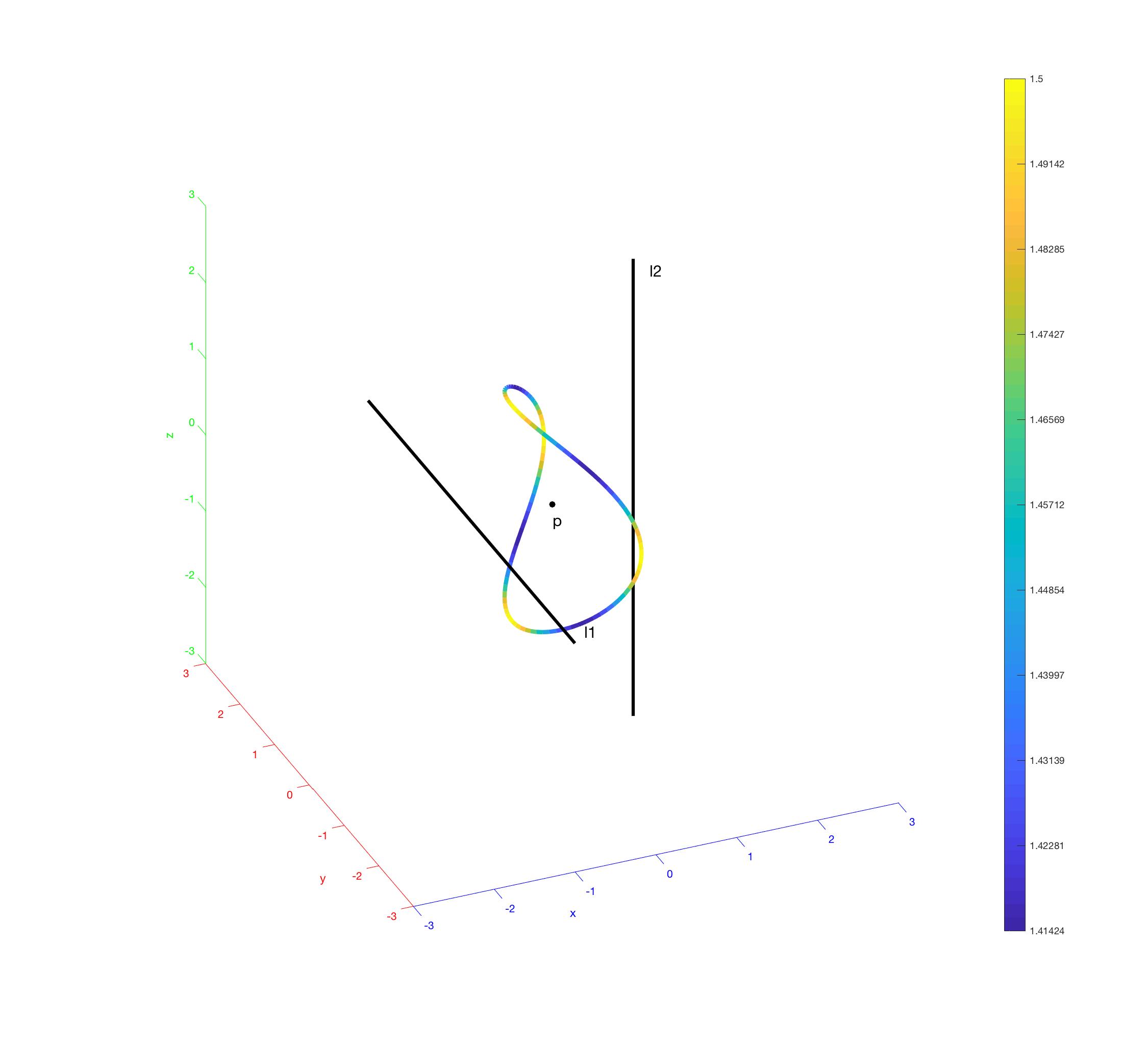}
	\caption{The trisector of 2 lines $l_1: x = -1 \wedge z = 0$ and $l_2: x = 1 \wedge y = 0$ and point $p = (0,0,0)$. 
	}\label{fig:bounded-trisector-segments}
\end{figure}


\section{Algorithmic Outlook and Related Aspects}\label{sec:outlook}

\subsection{Constructing the $\Gamma$-map from the Gaussian Map}\label{sec:gamma}

We show the construction of the $\Gamma$-map of $\fvd(L)$ from its Gaussian map. Recall that $\gm(\fvd(L))$ corresponds to the limit of $\Gamma(\fvd(L))$ as the radius of $\Gamma$ goes to infinity. 
However, $\gm(\fvd(L))$ has some additional vertices as a result of taking the limit to infinity. 
These are called ``vertices of anomaly''~\cite{Barequet2023} and correspond to directions that are orthogonal to each pair of lines. Let $\ell_1, \ell_2 \in L$ be two lines and $\overrightarrow{v}$ be a direction orthogonal to both.
The Gaussian map illustrates 
a vertex of anomaly at direction $\overrightarrow{v}$, even though only the farthest region of one line is unbounded in direction $\overrightarrow{v}$, we assume it to be $\ell_1$ w.l.o.g. Symmetrically only the region of $\ell_2$ is unbounded in the opposite direction $-\overrightarrow{v}$. 
Thus in order to capture the correct topology of the $\Gamma$-map, we need to
remove these vertices of anomaly; 
Figure~\ref{fig:verticesofanomaly} illustrates how this can be done locally around $\overrightarrow{v}$ and $-\overrightarrow{v}$. 

\begin{figure}[h]
	\centering
	\includegraphics[width=\linewidth]{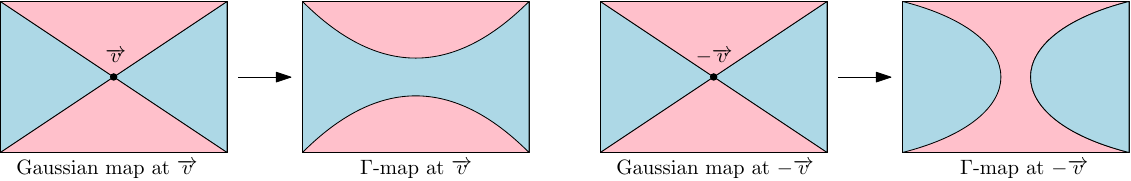}
	\caption{The local Gaussian map of $\fvd(L)$ around a pair of vertices of anomaly and the modifications to obtain $\Gamma(\fvd(L))$. The face of $\ell_1$ (resp.\ $\ell_2$) is shown in blue (resp.\ red).}\label{fig:verticesofanomaly}
\end{figure}

Given such a pair of vertices defined by lines $\ell_1$ and $\ell_2$, we can resolve them in constant time. First, we determine which line is farther in direction $\overrightarrow{v}$. We remove vertex $\overrightarrow{v}$ from the Gaussian map, merge its two incident faces of~$\ell_1$ into one, and reconnect the local boundary accordingly. Similarly for $-\overrightarrow{v}$. 
The Gaussian map has $2{n \choose 2}$ vertices of anomaly, hence, we can resolve them in $O(n^2)$ time and obtain the $\Gamma$-map.

\subsection{Consequences for the Farthest Voronoi Diagram}

Let $n_d, n_s, n_m, n_M$ denote the number of deletion, swap, local minimum, and local maximum events, respectively, during the collapse process for $\fvd(L)$.

\begin{proposition}\label{prop:combinatorialrelation}
The following numerical relation holds for the number of events: 
\(n_d + n_m - n_M = n^2 - n - 2\cdot \mathds{1}_d,\)
where $\mathds{1}_d$ is 0 if the collapse process ends with a deletion event and 1 otherwise.
\end{proposition}

\begin{proof}
We established that the collapse process for the farthest Voronoi diagram of lines starts from the $\Gamma$-map of the diagram, which has $2n^2-2n-4$ vertices by~\cite{papadopoulou2026}. By Figure~\ref{fig:events}, a deletion (resp.\ swap, local minimum, local maximum) event changes the number of vertices of the shrinking map by $-2$ (resp.\ $0$, $-2$, $+2$).
The only exception is when the last event of the collapse process is a deletion event. In this case, before the event there are 4 triangular faces and 4 vertices on the shrinking map, and after the event the map is empty. Hence, this special deletion event removes four vertices instead of two. 
The result follows by tracking the number of vertices during the collapse process.
\end{proof}

As a remark, by Proposition~\ref{prop:collapse}, a deletion event cannot be reversed during the collapse process. This is not true for swap events. In fact, consecutive swap events that are reverses of each other can happen (the first swaps a shrinking-map edge $(\ell_1, \ell_2)$ to $(\ell_3, \ell_4)$ and the second swaps back). In this case, the two Voronoi vertices and their incident Voronoi edges must form a special structure called a full twist, which is introduced in~\cite{papadopoulou2026}. 

By Theorem~9 of~\cite{papadopoulou2026}, the complexity of the farthest Voronoi diagram $|\fvd(L)|$ is bounded by $\Theta(n^2)$ plus its number of vertices. Together with Proposition~\ref{prop:combinatorialrelation}, we get the following.

\begin{corollary}\label{cor:fvd-complexity-events}
The farthest Voronoi diagram of $n$ lines $L$ has complexity
$\Theta(n^2 + n_s + n_M)$.
\end{corollary}

Next, we show some structural consequences of the collapse process. A local maximum event splits a face of the shrinking map locally into two, by Figure~\ref{fig:events}. The following can be proved by Proposition~\ref{prop:combinatorialrelation} and tracking face changes during the collapse process.
\begin{corollary}\label{cor:localmaxeventsplit}
	A local maximum event splits a shrinking-map face into two distinct faces.  
\end{corollary}

Finally, we characterize the topology of a 3D cell in $\fvd(L)$.

\begin{corollary}\label{cor:hole}
In the farthest Voronoi diagram of lines, every 3D cell $C$ has no holes. Consequently, for every sufficiently large ball \(B\), the truncation \(C\cap B\) is a topological 3-ball.
\end{corollary}
\begin{proof}
	Consider $\Gamma(\fvd(L))$. Its vertices and edges form a connected graph, which follows by induction starting from the base case shown in~\cite{papadopoulou2026}. Hence, there are no holes in the faces of $\Gamma(\fvd(L))$. By the ``no tunnel'' property shown in~\cite{Barequet2023}, every 3D cell $C$ of $\fvd(L)$ induces exactly one face $f$ on $\Gamma(\fvd(L))$. We trace the face of $C$ on the shrinking map during the collapse process. Assume that $C$ has a hole. Then two cases could happen during the collapse process: either a hole appears in some shrinking-map face, or two faces of the same cell merge. These two cases contradict items (3) and (2) of Proposition~\ref{prop:collapse}, respectively. Hence, there are no holes in $C$. The rest follows from Lemma~\ref{lem:raysinfreg}.
\end{proof}

\subsection{Algorithmic Outlook}

We outline an algorithm for constructing $\fvd(L)$ using the collapse
process. Start from $\Gamma(\fvd(L))$, which can be constructed in $O(n^2)$ time
as indicated in Section~\ref{sec:gamma}. 
%
During the whole process, we maintain a cell complex that is topologically
equivalent to the shrinking map, and a priority queue of event candidates.
Events are processed in decreasing priority order.

Each deletion, swap, and local minimum event is associated with the deletion of
at least one edge on the shrinking map, which  is used to identify the event and
compute its priority. Thus, these events are easy to identify and overall handle.

Every time 
an event happens, we need to 
update the two data structures. The shrinking map is updated topologically
according to Figure~\ref{fig:events}. A constant number of new potential
deletion, swap, or local minimum events appear or disappear after an event takes
place, and the priority queue is modified accordingly. 

Local maximum events are  not associated with any existing feature of the
shrinking map:
they correspond to pairs of boundary edges
of a common
face whose intersection defines a trisector extremum.
To identify the local maximum events we need to check pairs of boundary edges along 
each shrinking map face.
We can avoid checking all possible such pairs of edges due to  the following property.

\begin{lemma}\label{lem:crossing-candidates}
Consider a shrinking-map face $f$ of line $\ell_0$ whose neighbors are $\ell_1, \ldots, \ell_m$ in this order. Let $1\leq i< j< k< l\leq m$. Consider two pairs of edges $(\ell_0, \ell_i)$, $(\ell_0, \ell_k)$ and $(\ell_0, \ell_j)$, $(\ell_0, \ell_l)$.
If both pairs of edges meet at a local maximum of trisector $T(\ell_0, \ell_i,
\ell_k)$ and $T(\ell_0, \ell_j, \ell_l)$, respectively, then the
local maximum event with lower priority cannot occur during the collapse process. 
\end{lemma}

\begin{proof}
Consider the $\fvd(L')$, where $L'=\{\ell_0, \ell_i, \ell_j, \ell_k, \ell_l\}$, and the 3D cell $C\subseteq \freg(\ell_0, L')$ associated with the face $f$. WLOG, assume that the local maximum event on trisector $T(\ell_0, \ell_i, \ell_k)$ has higher priority.
Consider the 3D cell $C'\subseteq \freg(\ell_0, \{\ell_0, \ell_i, \ell_k\})$ that contains~$C$. By definition, such a 3D cell exists because $C$ can be obtained from $C'$ by trimming, after inserting lines $\ell_j, \ell_l$. In this cell $C'$, the local maximum event occurs, after which the shrinking map face is split into two. After inserting the lines $\ell_j, \ell_l$, they must lie on different shrinking-map faces of $\ell_0$ at the lower priority.
By Proposition~\ref{prop:collapse}, shrinking-map faces cannot merge. Hence, the local maximum event associated with $T(\ell_0, \ell_j, \ell_l)$ can never happen.
\end{proof}

In the spirit of Lemma~\ref{lem:crossing-candidates}, local maximum events that ``cross'' each other cannot both occur. Consequently, for each shrinking-map face $f$, we maintain a bucket $B_f$ of size $|f|$ that contains all possible local maximum events. When a deletion, swap, or local minimum event incident to $f$ happens, the number of neighbors of $f$ changes by $1$. In this case, $B_f$ can be updated by considering the interaction between the new neighbor and all others, in $O(|f|)$ time. When a local maximum event splits $f$ into two faces $f_1$ and $f_2$, we split $B_f$ into $B_{f_1}$ and $B_{f_2}$ accordingly, also in $O(|f|)$ time.

This gives a high-level algorithm in which the local maxima are generated on the
fly rather than all at once upfront. A direct implementation following this outline gives the time complexity $O(\left|\fvd(L) \right|k_{\max} \log n)+T_{\text{init}}$,
where $k_{\max}$ is the maximum size of a shrinking-map face. The term $T_{\text{init}}$ is the time for the initial bucket construction from $\Gamma(\fvd(L))$, which is $\sum_{f\in \Gamma(\fvd(L))}|f|^2$ if implemented by brute force.

In future research, we aim  
to obtain an output-sensitive algorithm whose running time
would not depend on $k_{\max}$. For that purpose,
instead of a bucket $B_f$, we could consider local maximum events associated
with a simplified \emph{Voronoi-like} structure of the neighbors of face $f$ in the spirit of~\cite{papadopoulou2023}. 
However, the algorithm of~\cite{papadopoulou2023} cannot be directly applied,
leaving the topic to future research.

\subsection{Generalization to Convex Sites under a Convex Distance Function}

It turns out that most statements can be extended to a set $S$ of $n$ disjoint convex sites in~$\R^3$ and a more universal distance measure.
Any bounded strictly convex set~$\C \subset \R^3$ which contains the origin (also called its center) in its interior can be used to define
a \textit{convex distance} between two points~$p,q \in
\R^3$: $d_\C(p,q) = \inf\limits_{t \geq 0} \Set{t \mid q \in p + t \cdot
	\C}$. 
The distance $d_\C(p,q)$ describes the amount~$t \geq 0$ by which~$\C$, when being placed
at~$p$, has to be scaled to cover $q$; see \cref{fig:scale}. 
The distance $d_\C(p,s)$ from a point~$p \in \mathbb{R}^3$ to a site~$s \in S$ is defined
as~\(d_\C(p,s) = \min\{d_\C(p,q) \mid q \in s\}\).

\begin{figure}[ht]
	\centering
	\includegraphics[width=0.4\linewidth]{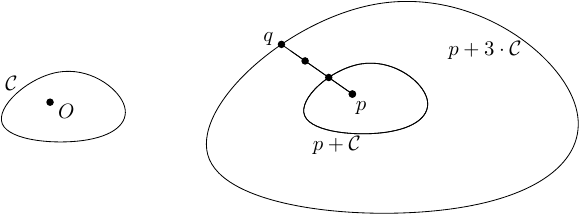}
	\caption{Convex distance induced by the convex set $\C$ containing the origin $O$: $d_\C(p,q) = 3$.}
	\label{fig:scale}
\end{figure}

The Lemmas/Theorems/Propositions/Corollaries \labelcref{lem:raysinfreg,cor:unboundedfarthestcells,thm:complete_event_list,prop:collapse,lem:stop-final,prop:combinatorialrelation,cor:localmaxeventsplit,lem:crossing-candidates}
still hold in this more general setting.

\bibliography{collapse_doi}



\end{document}